\documentclass[conference]{IEEEtran}

\usepackage{amsfonts}
\usepackage{amssymb}
\usepackage{stfloats}
\usepackage{graphicx}
\usepackage{subfigure}
\usepackage{amsmath}
\usepackage{array}
\usepackage{epstopdf}
\usepackage{graphicx}
\usepackage{amsthm}
\usepackage{algorithm}
\usepackage{algpseudocode}
\usepackage{amsmath}
\usepackage{graphics}
\usepackage{epsfig}

\usepackage[usenames]{color}

\newtheorem{Definition}{Definition}
\newtheorem{Problem}{Problem}
\newtheorem{Lemma}{Lemma}

\newtheorem{Policy}{Policy}

\newtheorem{Remark}{Remark}

\usepackage{geometry}
\begin{document}

\title{MDP-Based Scheduling Design for Mobile-Edge Computing Systems with Random User Arrival}

\IEEEoverridecommandlockouts{
\author{
	\IEEEauthorblockN {Shanfeng Huang\IEEEauthorrefmark{1}\IEEEauthorrefmark{2}, Bojie Lv\IEEEauthorrefmark{1}\IEEEauthorrefmark{3} and Rui Wang\IEEEauthorrefmark{1}\IEEEauthorrefmark{3}}

    \IEEEauthorblockA{
	\IEEEauthorrefmark{1}Department of Electrical and Electronic Engineering, Southern University of Science and Technology\\
	\IEEEauthorrefmark{2}Department of Electrical and Electronic Engineering, The University of Hong Kong\\
	\IEEEauthorrefmark{3}Peng Cheng Laboratory, Shenzhen, China\\
	Email: sfhuang@eee.hku.hk, lyubj@mail.sustech.edu.cn, wang.r@sustech.edu.cn}
	\thanks{This work was supported in part by National Natural Science Foundation of China under Grants 61771232 and 91838301, and in part by the Shenzhen Science and Technology Innovation Committee under Grant JCYJ20160331115457945.}
	}}

\maketitle

\begin{abstract}
	In this paper, we investigate the scheduling design of a mobile-edge computing (MEC) system, where the random arrival of mobile devices with computation tasks in both spatial and temporal domains is considered. The binary computation offloading model is adopted. Every task is indivisible and can be computed at either the mobile device or the MEC server. We formulate the optimization of task offloading decision, uplink transmission device selection and power allocation in all the frames as an infinite-horizon Markov decision process (MDP). Due to the uncertainty in device number and location, conventional approximate MDP approaches to addressing the curse of dimensionality cannot be applied. A novel low-complexity sub-optimal solution framework is then proposed. We first introduce a baseline scheduling policy, whose value function can be derived analytically. Then, one-step policy iteration is adopted to obtain a sub-optimal scheduling policy whose performance can be bounded analytically. Simulation results show that the gain of the sub-optimal policy over various benchmarks is significant.
\end{abstract}

\IEEEpeerreviewmaketitle

\section{Introduction}
With the proliferation of smart mobile devices, new applications with computation-intensive tasks are springing up, such as image recognition, online gaming and mobile augmented reality. Mobile-edge computing (MEC) is envisioned as a promising network architecture to address the conflict between resource-hungry applications and resource-limited devices.

MEC has been intensively investigated in recent years. In \cite{You2015SingleUserWPT}, the authors considered a single user MEC system powered by wireless energy transfer. The optimal offloading decision, local CPU frequency and time division between wireless energy transfer and offloading were derived in closed form via convex optimization theory. The authors in \cite{you2016multiuser} extended the work to a multi-user scenario and formulated the multi-user resource allocation problem as a convex optimization problem. An insightful threshold-based optimal offloading strategy was derived. Moreover, game-theory-based algorithms were designed to resolve the contention of multi-user MEC offloading decision in \cite{chen2016gametheorymec,chen2015decentalizedgame}.

The above works ignore the dynamics of mobile devices. Moreover, they assume the transmission and computation of a task can be finished within channel coherent time, which may not be the case in many applications. Considering the randomness of channel fading and task arrival, the scheduling in MEC systems becomes a stochastic optimization problem. Several works have been done to tackle such scheduling problems in MEC systems. In \cite{huang2012dynamic}, the authors considered a single-user MEC system and proposed a Lyapunov optimization algorithm to minimize the long-term average energy consumption. Also, the authors in \cite{mao2016power-delay} investigated the power-delay tradeoff of a multi-user MEC system via Lyapunov optimization. Moreover, the authors in \cite{liu2016delayopt} solved the power constrained delay-optimal task scheduling problem for an MEC system via MDP. Nevertheless, all these works consider either a single mobile device or a number of fixed mobile devices. The scheduling design with random arrival of mobile devices remains open.

In this paper, we would like to shed some light on the above issue. Specifically, we consider an MEC system, where a base station (BS) is connected with an MEC server. New mobile devices, each with a computation task, arrive randomly in the coverage region of the BS. Every computation task can be either computed locally or offloaded to the MEC server via uplink transmission. The optimization of task offloading decision, uplink device selection and power allocation in all the frames is formulated as an infinite-horizon MDP with discounted cost. Due to the dynamics of arrival and departure, the number of mobile devices in the MEC system is variable. The conventional approximate MDP approaches cannot be applied to address the curse of dimensionality, and a novel solution framework is proposed in this paper. Particularly, we first propose a baseline scheduling policy, whose value function can be derived analytically. Then, one-step policy iteration is applied to obtain the proposed sub-optimal policy. Since the value function for policy iteration can be calculated from analytical expression, the conventional value iteration can be avoided. Moreover, the analytical value function of the baseline policy provides an upper bound on the average discounted cost of the proposed sub-optimal policy.

\section{System Model}

\subsection{Network Model}
\begin{figure}[tb]
	\centering
	\includegraphics[scale=0.4]{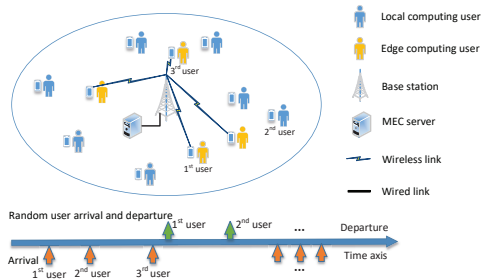}
	\caption{Illustration of MEC system model.}
	\label{fig: SystemModel}
\end{figure}

We consider an MEC system as illustrated in Fig. \ref{fig: SystemModel}, where a BS serves a region $ \mathcal C $ and an MEC server is connected with the BS. Mobile devices with computation tasks arrive randomly in the service region $ \mathcal C $. Binary computation offloading model is adopted, and every task is assumed to be indivisible in the sense of computing. Each task can be either computed locally or offloaded to the MEC server via uplink transmission.

There are a number of mobile devices in the cell region $ \mathcal C $, which may be quasi-static, moving inside or out of the region $ \mathcal C $. The mobile devices with computation tasks are named as active devices in the remainder of this paper. The time axis of computation and uplink transmission scheduling is organized by frames, each with a time duration $T_s$. In each frame, there is at most one new active device arrived in the cell with probability $ P_N\in(0,1] $. We have no restriction on the distribution and mobility model of the mobile devices in the cell. Instead, the distribution density of the new active device is represented as $ \lambda ({\mathbf l}) $ for arbitrary location in the cell region $ {\mathbf l} \in \mathcal C $. Thus, 
$\int_{\mathcal C} \lambda ({\mathbf l}) ds({\mathbf l}) = 1,$
and
$$\Pr [\mbox{New active device is in region } \mathcal C^{'}] \!\!= \!\!\int_{\mathcal C^{'}} \!\!\lambda ({\mathbf l}) ds({\mathbf l}), \ \forall \mathcal C^{'} \!\!\subseteq \mathcal \!C.$$
Moreover, it is assumed that the location of each active device is quasi-statistic in the cell when its task is being transmitted to the MEC server. The active devices become inactive when their computation tasks have been completed either locally or remotely at the MEC server.

Every new active device in the cell is assigned with a unique index. Let $ \mathcal U_L (t) $ and $ \mathcal U_E (t) $ be the sets of active devices in the $ t $-th frame whose tasks are computed locally and at the MEC server respectively, $ \mathcal{D}_{L}(t) \subseteq \mathcal{U}_L(t) $ and $\mathcal{D}_{E}(t) \subseteq \mathcal{U}_{E}(t)$ be the subsets of active devices whose computation tasks are accomplished in the $t$-th frame locally and at the MEC server respectively, $ n_t $ be the index of the new active device arriving at the beginning of $ t $-th frame.  If there is no active device arrival at the beginning of $ t $-th frame,  $ \{n_t\} = \emptyset $ where $ \emptyset $ represents the empty set. On the other hand, if there is a new active device arrival, the BS should determine if the computation task is computed at the device or the MEC server. Let $ e(t) \in \{0,1\} $ represents the decision, where $ e(t) = 1 $ means the task is offloaded to the MEC server and $ e(t)=0 $ means otherwise. Hence, the dynamics of active devices can be represented as
\begin{equation}
\mathcal{U}_E(t+1) = \left\{ \begin{array}{cc}
 \mathcal{U}_E(t) \cup \{n_{t}\} /\mathcal{D}_E(t)&  \mbox{when } e(t) = 1 \\
 \mathcal{U}_E(t)  /\mathcal{D}_E(t)&  \mbox{otherwise,}
\end{array}   \right.
\end{equation}
\begin{equation}
\mathcal{U}_L(t+1) = \left\{ \begin{array}{cc}
\mathcal{U}_L(t) \cup \{n_{t}\} /\mathcal{D}_L(t)&  \mbox{when } e(t)= 0 \\
\mathcal{U}_L(t)  /\mathcal{D}_L(t)&  \mbox{otherwise,}
\end{array}   \right.
\end{equation}
where operator $``/"$ denotes the set subtraction.

\subsection{Task Offloading Model}

The input data for each computation task is organized by segments, each with $ b_s $ information bits. Let $ d_k $ be the number of input segments for the task at the $ k $-th active device. It is assumed that the number of segments for each task is a uniformly distributed random integer between $ d_{min} $ and $ d_{max} $, i.e. $d_k\sim \mathbb U(d_{min}, d_{max})$. For the computation tasks offloaded to the MEC server, the input data should be delivered to the BS via uplink transmission. Hence, an uplink transmission queue is established at each active edge computing device. Let $ Q^E_k(t) $, $ \forall k \in \mathcal U_E(t) $, be the number of segments in the uplink transmission queue of the $ k $-th device at the beginning of the $ t $-th frame. Hence, $ \forall t \mbox{ with } \{n_t\} \neq \emptyset \mbox{ and } e(t)=1$,
\begin{equation*}
Q_{n_t}^E(t+1) = d_{n_t}.
\end{equation*}

In uplink, it is assumed that only one active device is selected in one uplink frame and the uplink transmission bandwidth is denoted as $W$.  Let $ H_k(t)=\sqrt{\rho_k(t)}h_k(t) ,\forall k\in \mathcal U_E(t) $, be the uplink channel state information (CSI) from the $ k $-th active device to the BS, where $ h_k(t) $ and $ \rho_k(t)  $ represent the small-scale fading and pathloss coefficients respectively. $ h_k(t) \sim \mathbb {CN} (0,1) $ is complex Gaussian distributed with zero mean and variance $ 1 $. Moreover, it is assumed that $ h_k(t) $ is independently and identically distributed (i.i.d.) for different $ t $ or $ k $. Let $ p_k(t) $ be the uplink transmission power of the $ k $-th active device if it is selected in the $ t $-th frame. The uplink channel capacity of the $ k $-th active device, if it is selected in the $ t $-th frame, can be represented by
\begin{equation*}
r_k(t) = W \log_2\left(1+\frac{p_k(t) \rho_k(t) |h_k(t)|^2}{\sigma_z^2}\right),
\end{equation*}
where $\sigma_z^2$ is the power of white Gaussian noise. Furthermore, the number of segments transmitted within the $t$-th frame can be obtained by
\begin{equation}
	\phi_k(t)=\biggl\lfloor \frac{r_k(t)T_s}{b_s} \biggr\rfloor,
\end{equation}
where $ \lfloor X \rfloor $ is the maximum integer less than or equal to $ X $. Hence, let $ a_t $ be the index of the selected uplink transmission device in the $ t $-th frame, we have the following queue dynamics for all $k\in \mathcal U_E(t)$,
\begin{equation}
	Q^E_k(t+1)=\left\{ \begin{array}{cc}
		\left[ Q^E_k(t)-\phi_k(t) \right]^+ & \mbox{if } k=a_t\\
		Q^E_k(t) & \mbox{if } k\neq a_t,
	\end{array} \right.
\end{equation}
where $ [X]^+ = \max \{0,X\} $.

As in many of the existing works \cite{You2015SingleUserWPT,mao2016dynamicmec,zhang2013mccstochastic}, it is assumed  that there are sufficiently many high-performance CPUs at the MEC server so that the computing latency at the MEC server can be neglected compared with the latency of local computing or uplink transmission. Moreover, due to relatively smaller sizes of computation results, the downloading latency of computation results is also neglected as in \cite{mao2016power-delay,mao2016dynamicmec,zhang2013mccstochastic}.

\subsection{Local Computing Model}
Following the computation models in \cite{you2016multiuser,mao2016power-delay}, the average number of CPU cycles for computing one bit of the input task data in the $k$ active device is denoted as $L_k$, which is determined by the types of applications. Denote the local CPU frequency of the $k$-th active device as $f_k$. We assume $L_k$ and $f_k$ are both uniformly distributed random variables, i.e. $L_k\sim \mathbb U(L_{min},L_{max})$ and $f_k\sim \mathbb U(f_{min},f_{max})$. An input data queue is established at each active local computing device. Let $ Q^L_k(t) $, $ \forall k \in \mathcal U_L(t) $, be the number of segments in the input data queue of the $ k $-th active device at the beginning of the $ t $-th frame. Hence, $ \forall t \mbox{ with } \{n_t\} \neq \emptyset \mbox{ and } e(t)= 0$,
\begin{equation*}
Q_{n_t}^L(t+1) = d_{n_t}.
\end{equation*}
Moreover, the queue dynamics at all active local computing devices can be written as
\begin{equation}
Q^L_k(t+1)= \left[Q^L_k(t)-\frac{f_k T_s}{L_k b_s}\right]^+, \ \forall k \in \mathcal U_L(t).
\end{equation}
Hence, the total computation time (measured in terms of frames) for $k$-th active device, whose task is computed locally, is given by
\begin{equation}
T_{loc}(d_k,f_k,L_k)=\biggl\lceil\frac{d_k b_s L_k}{f_k T_s}\biggr\rceil,
\end{equation}
where $  \lceil X\rceil $ is the minimum integer greater than or equal to $ X $. Moreover, the local computation power of $k$-th device is
\begin{equation}
p_{loc}(f_k)=\kappa f_k^3,
\end{equation}
where $\kappa$ is the effective switched capacitance related to the CPU architecture \cite{Burd1996Processor}.

\section{Problem Formulation}
In this section, we formulate the optimization of task offloading decision, uplink device selection and power allocation as an infinite-horizon MDP with discounted cost.

\subsection{System State and Scheduling Policy}

The system state and scheduling policy are defined as follows.

\begin{Definition}[System State]
	At the beginning of $ t $-th frame, the state of the MEC system is uniquely specified by $ \mathbf{S}_t =(\mathbf{S}^E_t,\mathbf{S}^L_t,\mathbf{S}^N_t)$, where
	\begin{itemize}
		\item $\mathbf S^E_t$ specifies the system status regarding the task offloading, including the set of edge computing devices $ \mathcal{U}_E(t) $, their uplink small-scale fading coefficients $ \mathcal{H}_E(t)\triangleq \{ h_k(t) | k \in \mathcal{U}_E(t) \} $ and pathloss coefficients $ \mathcal{G}_E(t)\triangleq \{  \rho_k(t) | k \in \mathcal{U}_E(t) \} $, and their uplink queue state information (QSI) $ \mathcal{Q}_E(t) \triangleq \{Q^E_k(t)| k\in \mathcal{U}_E(t)\}$.

		\item $\mathbf S^L_t$ specifies the system status regarding the local computing, including the set of local computing devices $\mathcal U_L(t)$, the application-dependent parameters $\mathcal {L}(t)\triangleq \{L_k(t) | k\in \mathcal U_L(t)\}$, their CPU frequencies $\mathcal F(t)\triangleq \{f_k(t)|k\in \mathcal U_L(t)\}$, and their QSI $\mathcal{Q}_L(t)\triangleq \{Q^L_k(t)|k\in \mathcal U_L(t)\}$.

		\item $\mathbf S^N_t$ specifies the system status regarding the new active device, including the indicator of new arrival $I_N(t)\triangleq I(\{n_t\} \neq \emptyset)$ where $I(\cdot)$ is the indicator function, its index $ n_t $, pathloss coefficient $ \rho_{n_t}(t) $, size of input data $ d_{n_t} ,$ CPU frequency $f_{n_t}$ and $L_{n_t}$.
	\end{itemize}

\end{Definition}

\begin{Definition}[Scheduling Policy]
The scheduling policy $ \Omega(\mathbf{S}_t)  \triangleq \left(a_t, p(t), e(t)\right)$ is a mapping from the system state $ \mathbf{S}_t $ to the scheduling actions, i.e, the index $a_t$ of the selected uplink transmission device in the $t$-th frame, the transmission power $p(t)$ and the offloading decision $e(t)$ for the new arriving active device (if any).
\end{Definition}

\subsection{Problem Formulation of MEC Scheduling}

According to Little's law, the average latency of one task is proportional to the average number of active devices in the system. Hence, we define the following weighted sum of the number of active devices and their power consumption as the system cost in the $ t $-th frame.
\begin{equation*}
g(\mathbf{S}_t, \Omega(\mathbf{S}_t))\! \triangleq\! |\mathcal{U}_E(t)| + |\mathcal{U}_L(t)|+ w [p(t) +\!\!\!\sum_{k\in\mathcal U_L (t)} p_{loc}(f_k)],
\end{equation*}
where $w$ is the weight on the power consumption of mobile devices. The overall minimization objective with the initial system state $ \mathbf{S}_1 $ is then given by
\begin{eqnarray}
\overline{G} (\Omega, \mathbf{S}_1)\!\!\!\!\!\! &\triangleq& \!\!\!\!\!\! \lim\limits_{T \rightarrow + \infty }\!\!\mathbb{E}_{\{\mathbf{S}^N_t,\mathcal{H}_E(t)|\forall t\}}\bigg[\sum_{t=1}^{T} \gamma^{t-1} g(\mathbf{S}_t, \Omega(\mathbf{S}_t)) \bigg| \mathbf{S}_1 \bigg], \nonumber
\end{eqnarray}
where $ \gamma $ is the discount factor. As a result, the MEC scheduling is formulated as the following infinite-horizon MDP.
\begin{Problem}[MEC Scheduling Problem]\label{prob:main}
\begin{eqnarray}
\Omega^{*}=\arg\min_{\Omega} \overline{G} (\Omega, \mathbf{S}_1).
\end{eqnarray}
\end{Problem}

According to \cite{Bertsekas2012Dynamic}, the optimal policy of Problem \ref{prob:main} can be obtained by solving the following Bellman's equations.
\begin{equation}
V\!(\mathbf{S}_t) \!\! =\!\! \min\limits_{\Omega(\mathbf{S}_t)}\!\! \bigg[\!g\!\left(\mathbf{S}_t,\Omega(\!\mathbf{S}_t\!)\!\right) + \!\!\sum_{\mathbf{S}_{t+1}}\!\! \gamma\! \Pr\!\left(\mathbf{S}_{t+1} | \mathbf{S}_{t}, \Omega(\mathbf{S}_t)\right) \!\!V\!(\mathbf{S}_{t+1})\! \bigg], \nonumber
\end{equation}
where $ V(\mathbf{S}) $ is the value function for system state $ \mathbf{S} $. Particularly, standard value iteration can be used to solve the value function, and the optimal policy can be derived by solving the minimization problem of the right-hand-side of the above Bellman's equations. In our problem, however, the traditional value iteration is intractable due to the following reasons: (1) the number of active devices is not fixed and the state space grows exponentially with the increasing number of active devices; (2) the spaces of small-scale fading and pathloss are continuous.

In order to address the above issues, we first reduce the system state space by exploiting (1) the independent distributions of small-scale fading and new active devices, and (2) the deterministic cost of local computing devices. The conclusion is summarized below.

\begin{Lemma}[Bellman's Equations with Reduced State Space] \label{lem: BE_reduced}
Define $C(n_t) \triangleq \sum_{\tau=1}^{T_{loc}(d_{n_t},f_{n_t},L_{n_t})}\!\gamma^{\tau} \left[ 1 + w p_{loc}(f_{n_t})\right]$, and $ {\widetilde{\mathbf{S}}_t} \triangleq \mathbf S^E_t/ \mathcal H_E(t)=(\mathcal U_E(t), \mathcal G_E(t), \mathcal Q_E(t)).$ Let
$$g'(\mathbf{S}_t,\Omega(\mathbf S_t))\triangleq|\mathcal{U}_E(t)| + w p(t) + I_N(t) (1-e(t)) C(n_t),$$
\begin{eqnarray}
W(\widetilde {\mathbf S})\!\!\!\!\!\! &\triangleq& \!\!\!\!\!\! \min_{\Omega} \lim\limits_{T \rightarrow + \infty }\mathbb{E} \bigg[\sum_{t=1}^{T} \gamma^{t-1} g'(\mathbf{S}_t, \Omega(\mathbf{S}_t)) \bigg| \mathbf{\widetilde S}_1=\mathbf{\widetilde S} \bigg]. \nonumber
\end{eqnarray}
The optimal scheduling action for the system state $ \mathbf{S}_t $, denoted as $ \Omega^{*} (\mathbf{S}_t) $, can be obtained as follows.
\begin{eqnarray}
\Omega^{*} (\mathbf{S}_t) &=&\arg\min\limits_{\Omega(\mathbf{S}_t)}\!\! \bigg\{g'\!\left( \mathbf{S}_t,\Omega(\mathbf{S}_t)\right)\!\!\nonumber\\
&+&\!\!\!\!\! \sum_{\widetilde{\mathbf{S}}_{t+1}}\!\gamma \!\Pr\!\left(\widetilde{\mathbf{S}}_{t+1} | \mathbf{S}_{t}, \Omega(\mathbf{S}_t)\right)\!\!W \!(\widetilde{\mathbf{S}}_{t+1}) \!\bigg\}.
\label{eqn:reduced-problem}
\end{eqnarray}
\end{Lemma}

\begin{proof}
	Please refer to appendix A.
\end{proof}

\section{Low-Complexity Solution}

In this section, we first introduce a heuristic scheduling policy as the baseline policy, whose value function can be derived analytically. Then, the proposed low-complexity sub-optimal policy can be obtained via the above value function and one-step policy iteration. The derived value function becomes the cost upper bound of the proposed policy.

\subsection{Baseline Scheduling Policy}
The baseline scheduling policy is elaborated below.
\begin{Policy}[Baseline Scheduling Policy $\Pi$]\label{pol:baseline}
	Given the system state $\mathbf{S}_t$, the baseline scheduling policy $\Pi(\mathbf{S}_t)=\left(a_t, p(t), e(t)\right)$ is provided below
	\begin{itemize}
		\item Uplink transmission device selection $a_t = \min \mathcal{U}_E(t)$, $\forall t$. Thus, the BS schedules the uplink device in a first-come-first-serve manner.
		\item The
		transmission power $p(t)$ compensates the large-scale fading (link compensation). Thus,
		\begin{align}
		p(t)=\frac{p_r}{\rho_{a_t}}, \forall t,
		\end{align}
		where $p_r$ is the average receiving power at the BS.
		\item The task of the new active device is offloaded to MEC server only when there is currently no active edge computing device, i.e.,
		\begin{align}
		e(t)=I\bigg(\mathcal{U}_E(t)=\emptyset\bigg), \forall t.
		\end{align}
	\end{itemize}
\end{Policy}

 Given system state $ \widetilde{\mathbf{S}}$, the value function  of policy $\Pi$ is defined as
\begin{align}
W_{\Pi}(\widetilde{\mathbf{S}})\!\triangleq\!\!\lim\limits_{T \rightarrow + \infty }\!\mathbb{E}\bigg[\! \sum_{t=1}^{T} \!\gamma^{t-1} g'\left(\mathbf{S}_t,\Pi(\mathbf{S}_t)\right)\bigg|\widetilde{\mathbf{S}}_1=\widetilde{\mathbf{S}}\bigg].
\end{align}
Denote the index of the $ k $-th arrived active edge computing device in $ \mathcal{U}_E \in {\widetilde{\mathbf S }}$  as $ m_{k} $. Let $ T_{k} $ be the number of frames for completing the uplink transmission of the $ m_{k} $-th device. $W_{\Pi}(\widetilde{\mathbf{S}})$ can be written as
\begin{align}
W_{\Pi}(\widetilde{\mathbf{S}})\!\! =  \!\mathbb{E}_{\{T_{k}|\forall k\}}\!\left [ \sum_{ k =1}^{|\mathcal{U}_E|}\!\! \left(\! w \gamma^{\sum\limits_{\!i\!=\!1\!}^{\!k\!-\!1\!}T_{i}} \frac{1\!-\!\gamma^{T_{\!k\!}}}{1\!-\!\gamma} \!\frac{p_r}{\rho_{\!m_{k}}}\!\! +\! \frac{1\!\!-\!\gamma^{\sum\limits_{i\!=\!1\!}^{k}T_{\!i}}}{1-\gamma}\!\right)\!\right ]&\nonumber\\
+P_N \mathbb{E}_{\{T_{k}|\forall k\},\{C(n_t)|\forall t\leq \sum_{k=1}^{|\mathcal{U}_E|} T_k \}}\bigg[\sum_{t=1}^{\sum_{k=1}^{|\mathcal{U}_E|}T_{k}} \gamma^{t-1}C(n_t) \bigg]&
\nonumber\\
\!+ \!\!\lim\limits_{\!T\!\rightarrow+\infty}\!\!\!\mathbb{E}_{\{T_{k}|\forall k\}, \{\mathbf{S}^N_t|\forall t> \sum_{k=1}^{|\mathcal{U}_E|} T_k  \}}\!\!\bigg[ \!\!\sum_{t\!=\!1\!+\!\sum_{k=\!1}^{|\mathcal{U}_E|}\!T_{k}}^{T}\!\!\!\!\! \gamma^{t-1} \!g'(\mathbf{S}_t,\Pi(\mathbf{S}_t)\!)\!\bigg],&\label{eqn:accurate_W_pi}
\end{align}
where the first term is the average offloading cost of the existing active edge computing devices, the second term is the average local computing cost from the first frame to the $ (\sum_{k=1}^{|\mathcal{U}_E|}T_{k}) $-th frame, and the last term is the average cost after the $ (\sum_{k=1}^{|\mathcal{U}_E|}T_{k}) $-th frame. The first two terms can be calculated by noticing the following factors.
\begin{itemize}
	\item Since the amount of input data of one task is usually much larger than the throughput of one frame, we have the following approximation
	\begin{align}
	T_{k} \approx \left\lceil\frac{ Q_{m_{k}} b_s }{\mathbb{E}_{
			h} W\log_2 \left( 1 + \frac{p_r |h|^2}{\sigma_z^2} \right)T_s} \right\rceil, \forall k,
	\end{align}
	where $ \mathbb{E}_{h}  $ is the expectation w.r.t. small-scale fading.

	\item
	$\mathbb{E} [C(n_t)]=\frac{\sum_{d_{min}}^{d_{max}}\int_{f_{min}}^{f_{max}}\int_{L_{min}}^{L_{max}}C(n_t)dfdL}{(d_{max}-d_{min}+1)(f_{max}-f_{min})(L_{max}-L_{min})}.$
\end{itemize}

Define the state without any edge computing device as $ \widetilde{\mathbf{S}}^{*} \triangleq [\mathcal{U}_E=\emptyset, \mathcal{G}_E=\emptyset, \mathcal{Q}_E=\emptyset]$. The third term of (\ref{eqn:accurate_W_pi}) can be written as $ \gamma^{\sum_{k=\!1}^{|\mathcal{U}_E|}\!T_{k}} W_{\Pi} (\widetilde{\mathbf{S}}^{*}) $, whose expression is derived in the following lemma.

\begin{Lemma}[Analytical Expression of $W_{\Pi}(\widetilde{\mathbf{S}}^*)$]\label{lem: ExpW_Pi_s1}
	Let $ \mathbf{u} = [1 \ 0 \ 0 \ 0 \ ...]^{\mathbf{T}} \in \mathbb R^{(d_{max}+1)\times 1}$ be the vector whose elements are all $ 0 $ except the first one, and $ \mathbf{a}^{\mathbf{T}}$ be the transpose of vector $ \mathbf{a} $. Let $ \mathbf{g} = [g_1 \ g_2 \ ... \ g_{d_{max}+1}]^{\mathbf{T}} \in \mathbb R^{(d_{max}+1)\times 1}$, where $g_1=0$, $g_i=1+w\mathbb{E}_{\rho_{n_t}}[\frac{p_r}{\rho_{n_t}}]+P_N\mathbb{E}[C(n_t)]$,  $\forall i=2,3,...,d_{max}+1$. The analytical expression of $W_{\Pi}(\widetilde{\mathbf{S}}^{*} )$ is given by
	\begin{align} \label{eq:W_Pi_S1}
		W_{\Pi}(\widetilde{\mathbf{S}}^{*} )=\sum_{t=1}^{+\infty}\mathbf{u}^{\mathbf T} (\gamma \mathbf{M})^{t-1}\mathbf{g}=\mathbf{u}^{\mathbf T} (\mathbf{I}-\gamma \mathbf{M})^{-1}\mathbf{g},
	\end{align}
where $\mathbf{I}\in \mathbb{R}^{(d_{max}+1)\times (d_{max}+1)}$ is the identity matrix. Moreover, the elements of the transition probability matrix $\mathbf{M}\in \mathbb{R}^{(d_{max}+1)\times (d_{max}+1)}$ are given by
\begin{itemize}
	\item  $\mathbf{M}_{1,1}=1-P_N$,
	\item $\mathbf{M}_{1,j}=0$, for $j=2,3,...,d_{min}$,
	\item $\mathbf{M}_{1,j}=\frac{P_N}{d_{max}-d_{min}+1}$, for $j=d_{min}+1,d_{min}+2,...,d_{max}+1$,
	\item $\mathbf{M}_{i,j}=0$, for $1<i< j$,
	\item  $\mathbf{M}_{i,1}=  \exp\{-\frac{[2^{\frac{(i-1) b_s}{WT_s}}-1]\sigma^{2}_{z}}{p_r}\}$, for $i=2,3,...,d_{max}+1$,
	\item $\mathbf{M}_{i,j}\!\!=\!\!\exp\{\!-\frac{[2^{\frac{(i-j)b_s}{WT_s}}-1]\sigma^{2}_{z}}{p_r}\}\!-\!\exp\{-\frac{[2^{\frac{(i-j+1)b_s}{WT_s}}-1]\sigma^{2}_{z}}{p_r}\}$, otherwise.

\end{itemize}
\end{Lemma}

\begin{proof}
Please refer to Appendix B.
\end{proof}

\subsection{Scheduling with Approximate Value Function}

In this part, we use the value function of the baseline policy $ W_{\Pi}(\widetilde {\mathbf S}) $ to approximate the value function of the optimal policy $ W(\widetilde{\mathbf S}) $ in optimization problem (\ref{eqn:reduced-problem}). Hence, in the $ t $-th frame, the scheduling actions can be derived by the following problem.

\begin{Problem}[Sub-optimal Scheduling Problem]\label{prob: sub-optimal}
	\begin{align}
	\Pi'&=\arg\min\limits_{\Omega({\mathbf{S}}_t)}  \bigg\{g'\left(\mathbf{S}_t,\Omega(\mathbf{S}_t)\right) \nonumber \\
	 &+\sum_{\widetilde{\mathbf{s}}_{t+1}} \gamma \Pr\left(\widetilde{\mathbf{S}}_{t+1} | \mathbf{S}_{t}, \Omega(\mathbf{S}_t)\right) W_{\Pi}(\widetilde{\mathbf{S}}_{t+1}) \bigg\}.
	\end{align}
\end{Problem}

Problem \ref{prob: sub-optimal} can be solved by the following steps.
\begin{itemize}
\item {\bf Step 1:} For each $k\in \mathcal{U}_E(t)$, calculate
\begin{align}
G_E^{k}&=	\min\limits_{p_k(t)}  \bigg\{wp_k(t)\nonumber \\
		&+ \!\!\sum_{\widetilde{\mathbf{S}}_{t+\!1}} \!\!\gamma \!\Pr\!\left(\widetilde{\mathbf{S}}_{t+\!1} \!| \mathbf{S}_{t}, \!e(t)\!\!=\!\!1, \!a_t\!\!=\!\!k, \!p_k\!(t)\!\right)\!\! W_{\Pi}(\widetilde{\mathbf{S}}_{t+\!1\!}) \!\bigg\},\nonumber
\end{align}
and
\begin{align}
			G_L^{k}&=C(n_t)+\min\limits_{p_k(t)}  \bigg\{wp_k(t)\nonumber \\
			&+ \!\!\sum_{\widetilde{\mathbf{S}}_{t+\!1}} \!\!\gamma \!\Pr\!\left(\widetilde{\mathbf{S}}_{t+\!1} \!| \mathbf{S}_{t}, \!e(t)\!\!=\!\!0, \!a_t\!\!=\!\!k, \!p_k\!(t)\!\right)\!\! W_{\Pi}(\widetilde{\mathbf{S}}_{t+\!1\!}) \!\bigg\}.\nonumber
\end{align}
Let $ p_{k,E}^{*} (t)$ and $ p_{k,L}^{*} (t)$ be the optimal power allocation of the above two problems respectively. Note that if there is no arrival of new active device, $ C(n_t) = 0 $, and the above two problems are the same.

\item {\bf Step 2:} If $\min_k G_{E}^k<\min_k G_{L}^k$,  the solution of Problem \ref{prob: sub-optimal} is given by $(e(t)=1,a_t=k_E^*,p_k(t)=p_{k_E^*,E}^*(t))$, where $k_E^*=\arg\min_k G_{E}^k$. Otherwise, the solution of Problem \ref{prob: sub-optimal} is given by	$(e(t)=0,a_t=k_L^*,p_k(t)=p_{k_L^*,L}^*(t))$, where $k_L^*=\arg\min_k G_{L}^k$.
	\end{itemize}

Moreover, we have the following bounds on the proposed scheduling policy.

\begin{Lemma}[Performance Bounds]\label{lem: PerformanceBounds}
	Let $ W_{\Pi'} (\widetilde{\mathbf{S}})\!\triangleq\!\!\lim\limits_{T \rightarrow + \infty }\!\mathbb{E}\bigg[\! \sum_{t=1}^{T} \!\gamma^{t-1} g'\left(\mathbf{S}_t,\Pi'(\mathbf{S}_t)\right)\bigg|\widetilde{\mathbf{S}}_1=\widetilde{\mathbf{S}}\bigg]$ be the value function of the policy $\Pi'$, then
	\begin{align}
		W(\widetilde{\mathbf{S}}_t) \leq W_{\Pi'}(\widetilde{\mathbf{S}}_t) \leq	W_{\Pi}(\widetilde{\mathbf{S}}_t), \forall \widetilde{\mathbf{S}}_t.
	\end{align}
\end{Lemma}
\begin{proof}
Since policy $\Pi'$ is not the optimal scheduling policy, $W(\widetilde{\mathbf{S}}_t) \leq	W_{\Pi'}(\widetilde{\mathbf{S}}_t)$ is straightforward. The proof of $W_{\Pi'}(\widetilde{\mathbf{S}}_t) \leq W_{\Pi}(\widetilde{\mathbf{S}}_t)$ is similar to the proof of \emph{ Policy Improvement Property} in chapter II of \cite{Bertsekas2012Dynamic}.
\end{proof}

\begin{Remark}
In most of the existing literature on wireless resource allocation with approximate MDP \cite{RWang2013RelayApproxMDP,YSun2019PushingCaching,BLv2019Cache}, the performance can hardly be bounded analytically. The novel solution framework proposed in this paper provides a low-complexity policy whose performance can be bounded analytically. Particularly, since the approximate value function is derived analytically, the conventional value iteration can be avoided, which reduces the computation complexity. Moreover, since the proposed policy is obtained from the baseline policy with analytical value function, its average cost is naturally upper-bounded by the value function of the baseline policy.
\end{Remark}

\section{Simulation Results}
In this section, we evaluate the performance of the proposed low-complexity sub-optimal scheduling policy by numerical simulations. In the simulation, the frame duration is $T_s=10$ ms. The input data size of each task is uniformly distributed between $ 200 $ and $ 300 $ segments, each of a size of $ 10 $ Kb. Local CPU frequency is $ 1 $GHz and $ 500 $ CPU cycles are needed to compute 1-bit input data. The effective switched capacitance is $\kappa=10^{-28}$. Moreover, the uplink bandwidth is $W=10 \ \mbox{MHz} $, noise power is $\sigma_z^2=-104 \ \mbox{dBm} $. We compare our proposed scheduling policy with three benchmark policies including (1) the baseline policy (BSL) as elaborated in section IV-A; (2) \emph{all local computing policy} (ALC), where all the active devices execute their tasks locally; and (3) \emph{all edge computing policy} (AEC), where all the active devices offload their tasks to the MEC server.

\begin{figure}[tb]
	\centering
	\includegraphics[height=3.815cm,width=5.6cm]{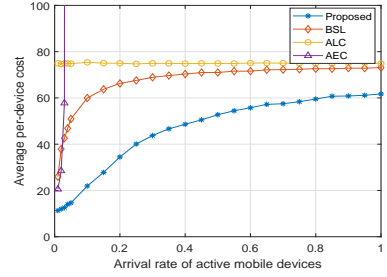}
	\caption{Average per-device cost versus arrival rate}
	\label{fig:cost}
\end{figure}

\begin{figure}[tb]
	\centering
	\includegraphics[height=3.815cm,width=5.6cm]{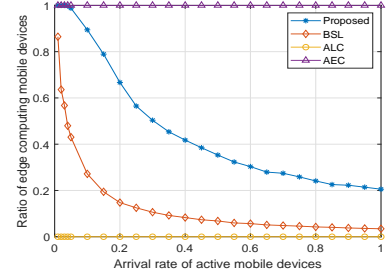}
	\caption{Ratio of edge computing versus arrival rate}
	\label{fig:ratio}
\end{figure}

Fig. \ref{fig:cost} shows the average per-device costs versus the arrival rates of active devices. It can be observed that the average per-device costs of all the policies grow with the increase of arrival rate except ALC policy. For ALC policy, since all the active devices computed their tasks locally, the arrival rate has no influence on the average per-device cost. For AEC policy, the average per-device cost grows quickly with the increase of arrival rate due to limited wireless transmission capability. It is also shown that our proposed policy always outperforms  BSL policy especially when the arrival rate falls in the region of $ (0,0.4) $. Besides, it can be seen that when the arrival rate is sufficiently large, the costs of both BSL policy and our proposed policy converge to the cost of ALC policy. This observation can be explained by Fig. \ref{fig:ratio} which shows that the ratio of edge computing devices tends to $ 0 $ for sufficiently large arrival rate. Moreover, as shown in Fig. \ref{fig:ratio}, the ratio of edge computing devices of our proposed policy is remarkably lager than that of BSL policy. Hence, our proposed policy can better exploit the MEC server to save the energy consumption of mobile devices and reduce latency.

\section{Summary}
In this paper, we formulate the scheduling design of a multi-user MEC system as an infinite-horizon MDP, and propose a novel low-complexity solution framework to obtain a sub-optimal policy via an analytical approximation of value function. The performance of the sub-optimal policy can be analytically bounded. Simulation results demonstrate the significant performance gain of the proposed scheduling policy over various benchmarks.

\section*{Appendix A: Proof of Lemma \ref{lem: BE_reduced}}

Due to limited space, we only provide the sketch of the proof. Since $\lim\limits_{T \rightarrow + \infty }\mathbb{E} \bigg[\sum_{t=1}^{T} \gamma^{t-1} g'(\mathbf{S}_t, \Omega(\mathbf{S}_t)) \bigg| \mathbf{\widetilde S}_1 \bigg]=\lim\limits_{T \rightarrow + \infty }\mathbb{E}\bigg[\sum_{t=1}^{T} \gamma^{t-1} g(\mathbf{S}_t, \Omega(\mathbf{S}_t)) \bigg| \mathbf{S}_1 \bigg],$ minimizing the right-hand-side is equivalent to minimizing the left-hand-side. The conclusion of Lemma \ref{lem: BE_reduced} can be obtained by writing down the Bellman's equations for  $\min_{\Omega}\lim\limits_{T \rightarrow + \infty }\mathbb{E} \bigg[\sum_{t=1}^{T} \gamma^{t-1} g'(\mathbf{S}_t, \Omega(\mathbf{S}_t)) \bigg| \mathbf{\widetilde S}_1 \bigg]$, and taking expectation w.r.t. small-scale fading and $ \mathbf{S}_t^N $.

\section*{Appendix B: Proof of Lemma \ref{lem: ExpW_Pi_s1}}
With baseline policy $\Pi$ and initial system state $\widetilde{\mathbf S}^*$, there is at most one edge computing device. In fact, the $ i $-th element of vector $ \mathbf{u} $ represents the probability that there is one edge computing device with $ (i-1) $ segments in the uplink transmission queue; the $ (i,j) $-th element of matrix $ \mathbf{M} $ represents the probability that there is one edge computing device with $ (j-1) $ segments in the uplink transmission queue in the next frame, given $ (i-1) $ segments in the current frame. Hence, we have the following discussion on $ \mathbf{M}_{i,j} $.
\begin{itemize}
	\item  $i=1, j=1$: Transition from 1st state (0 segment) to 1st state means that there is no new active device arrival. Hence  $M_{1,1}=1-P_N$.
	\item $i=1, j=2,3,...,d_{min}$: Since the minimum size of a new task is $d_{min}$ segments, it is impossible to transit from $ 0 $ segment to  $(j-1)$ segments. Hence, $M_{i,j}=0$.
	\item  $i=1, j=d_{min}+1,d_{min}+2,...,d_{max}+1$: This means there is a new active device arrival. The probability of a new active device arrival is $P_N$ and the task size of the new active device is uniformly distributed between $d_{min}$ to $d_{max}$. Thus, the probability of transiting from 1st state (0 segment) to $j$-th state ($j-1$ segments) for $ j=d_{min}+1,d_{min}+2,...,d_{max}+1$ is $M_{i,j}=\frac{P_N}{d_{max}-d_{min}+1}$.
	\item $1<i< j\leq d_{max}+1$: $i>1$ indicates that the current uplink transmission queue is not empty. Hence, the edge computing queue will not increase since the new arrival will be scheduled for local computing under policy $\Pi$. Therefore, $\mathbf{M}_{i,j}=0$, for $\forall 1<i< j$.
	\item  $i=2,3,...,d_{max}+1, j=1$: This means that the current edge computing device will finish transmitting the remaining ($i-1$) segments within current frame. Hence, $\mathbf{M}_{i,1}= \mbox{Pr}\left(W\log_2(1+\frac{p_r |h|^2}{\sigma_z^2}\geq (i-1)b_s)\right) =\exp\{-\frac{[2^{(i-1) b_s/(WT_s)}-1]\sigma^{2}_{z}}{p_r}\}$.
	\item $d_{max}+1\geq i \geq j>1$: This means that the edge computing device will transmit ($i-j$) segments within current frame. Hence, $\mathbf{M}_{i,j}=\mbox{Pr}\left((i-j)b_s\leq W\log_2(1+\frac{p_r |h|^2}{\sigma_z^2}\leq (i-j+1)b_s\right)=\exp\{\!-\!\frac{[2^{(i\!-\!j)b_s\!/(\!W\!T_s\!)\!}\!-\!1\!]\sigma^{2}_{z}}{p_r}\!\}-\exp\{-\frac{[2^{(i-j+1)b_s/(WT_s)}-1]\sigma^{2}_{z}}{p_r}\}$.

\end{itemize}

\bibliographystyle{IEEEtran}
\bibliography{reference_shortened}

\end{document}